\documentclass[12pt]{article}

\usepackage{graphicx, amssymb, latexsym, amsfonts, amsmath, lscape, amscd,
amsthm, color, epsfig, mathrsfs, tikz, enumerate}
\usepackage{soul}

\newtheorem{theorem}{Theorem}[section]

\newtheorem{corollary}[theorem]{Corollary}

\newtheorem{lemma}[theorem]{Lemma}
\newtheorem{proposition}[theorem]{Proposition}

\newtheorem{definition}[theorem]{Definition}

\newtheorem{observation}{Observation}

\newcommand\DELETE[1]{}

\begin{document}

\title{On a special class of boxicity 2 graphs}
\author{
{\sc Sujoy Kumar Bhore$^{(a)}$}, {\sc Dibyayan Chakraborty$^{(b)}$}, {\sc Sandip Das$^{(b)}$}, \\{\sc Sagnik Sen$^{(c)}$}\\
\mbox{}\\
{\small $(a)$ Ben-Gurion University, Beer-Sheva, Israel}\\
{\small $(b)$ Indian Statistical Institute, Kolkata, India}\\
{\small $(c)$ Indian Statistical Institute, Bangalore, India}
}

\date{}

\maketitle

\maketitle
\begin{abstract}
 We define and study a class of graphs, called 2-stab interval graphs (2SIG),  
with boxicity 2 which properly contains the class of interval graphs.
A 2SIG is an axes-parallel rectangle intersection graph where the rectangles have unit height (that is, length of the side parallel to $Y$-axis) and intersects either of the two 
fixed lines, parallel to the $X$-axis, distance $1+\epsilon$ ($0 < \epsilon < 1$) apart. Intuitively, 2SIG is a 
graph obtained by putting some edges between two interval graphs in a particular rule.
It turns out that for these kind of graphs, the chromatic number of any of its induced subgraphs is bounded by 
twice of its (induced subgraph) clique number.
This shows that the graph, even though not perfect, is not very far from it.
Then we prove similar results for some subclasses of 2SIG and provide efficient algorithm for finding their clique
number. We provide a matrix characterization for a subclass of 2SIG graph.
 \end{abstract}


\section{Introduction}
A \textit{geometric intersection graph}~\cite{golumbic} is a graph whose vertices are represented by geometric objects and two vertices are adjacent 
 if their corresponding geometric objects  intersect. 
\textit{Boxicity}~\cite{krato2}  of  a graph $G$ is the minimim $k$ such that $G$ can be expressed as a geometric intercestion graph of 
 of  axes-parallel $k$ dimensional rectangles. The class of \textit{boxicity $k$ graphs} is the class of graphs with boxicity at most $k$. 
The class of graphs with boxicity $1$ is better known as
 \textit{interval graphs}~\cite{golumbic}
 (intersection  of  real intervals) 
 while the class of graphs with boxicity $2$ is better known as 
 \textit{rectangle intersection graphs}~\cite{krato2} 
 (intersection of axes-parallel 
 rectangles). 


It is known that several questions 
(for example, recognition, determining clique number, determining chromatic number) 
that are $NP$-hard in general becomes polynomial time solvable when restricted to the class of 
interval graphs 
 while 
they remain $NP$-hard for the family of graphs with boxicity $k$ (for $k \geq 2$)~\cite{krato2}.
The reason for this dichotomy is probably because interval graphs are \textit{perfect} 
(defined in Section~\ref{preliminaries})
 while boxicity $k$ (for $k \geq 2$) graphs are not perfect
(those questions are polynomial time solvable for perfect graphs as well)~\cite{golumbic}. 

Naturally we are interested in exploring the objects that lie in between, that is, 
the proper subclasses of graphs with boxicity 2 that contains all interval graphs. 
Several such graph classes have been defined and studied~\cite{zhang}~\cite{hell}~\cite{brand}. 
In this article, we too define such a graph class and study its different aspects. 
We keep in mind that `perfectness' is probably the key word here. 
Our class of graphs is not perfect but it contains all interval graphs and is a proper subclass of boxicity 2 graphs. 
Moreover, our graph class is based on local structures of boxicity 2 graphs in some sense. 
Thus, the study of this class may help us understand the structure of boxicity 2 graphs in a better way. 

As a matter of fact, the definition of our graph class is motivated from the definition of 
a well-known class of perfect graphs, the split graphs. 
A \textit{split graph} is obtainted by putting edges between a clique and a set of independent vertices~\cite{golumbic}. 
Note that a complete graph and an independent set are the two  extreme trivial examples of perfect graphs. 
So when we put  edges between these two types of perfect graphs,  what we obtain
is again perfect. 

Motivated by this example, we wondered what would happen if we put edges between other kinds of perfect graphs. 
We take two interval graphs and put edges in between, following a particular rule. 
What we  obtain is a class of geometric intersection graphs, not perfect, with certain properties which enables us to call them ``nearly perfect''~\cite{gyar}. That is, the chromatic number of each induced subgraph is bounded by a function of its clique number; 
a linear function in our case.

 Let $y = 1$ be the \textit{lower stab line}  and 
$y = 2+ \epsilon $ be the \textit{upper stab line} where $ \epsilon \in (0,1)$  is a constant. 
Now consider
axes-parallel 
 rectangles with unit height (length of the side parallel to $Y$-axis) that intersects one of the stab lines.
 A  \textit{2-stab interval graph (2SIG)} is a graph $G$ that can be represented as an intersection graph of such rectangles.
Such a  representation $R(G)$ of $G$ is called a \textit{2-stab  representation}  (for example, see Fig.~\ref{bridge}). 
 A 2SIG may have more than one 2-stab representation.

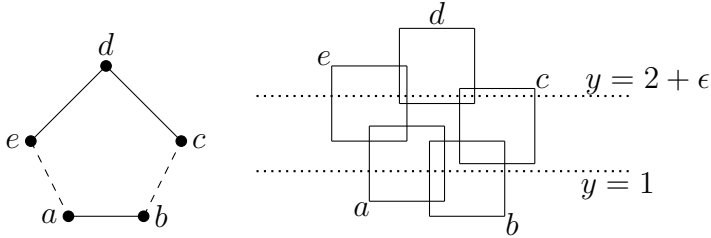
\begin{figure}

\centering
\begin{tikzpicture}

\filldraw [black] (0,0) circle (2pt) {node[left]{$a$}};
\filldraw [black] (1,0) circle (2pt) {node[right]{$b$}};
\filldraw [black] (1.5,1) circle (2pt) {node[right]{$c$}};
\filldraw [black] (-.5,1) circle (2pt) {node[left]{$e$}};
\filldraw [black] (.5,2) circle (2pt) {node[above]{$d$}};

\draw[-] (0,0) -- (1,0);
\draw[-] (-.5,1) -- (.5,2);
\draw[-] (1.5,1) -- (.5,2);

\draw[dashed] (-.5,1) -- (0,0);
\draw[dashed] (1,0) -- (1.5,1);


\draw[-] (4,.2) -- (5,.2);
\draw[-] (4,1.2) -- (5,1.2);
\draw[-] (4,.2) -- (4,1.2);
\draw[-] (5,.2) -- (5,1.2);

\node at (3.9,.1) {$a$};

\draw[-] (4.8,0) -- (5.8,0);
\draw[-] (4.8,1) -- (5.8,1);
\draw[-] (4.8,0) -- (4.8,1);
\draw[-] (5.8,0) -- (5.8,1);

\node at (5.9,-.1) {$b$};

\draw[-] (3.5,1) -- (4.5,1);
\draw[-] (3.5,2) -- (4.5,2);
\draw[-] (3.5,1) -- (3.5,2);
\draw[-] (4.5,1) -- (4.5,2);

\node at (3.4,2.1) {$e$};

\draw[-] (4.4,1.5) -- (5.4,1.5);
\draw[-] (4.4,2.5) -- (5.4,2.5);
\draw[-] (4.4,1.5) -- (4.4,2.5);
\draw[-] (5.4,1.5) -- (5.4,2.5);

\node at (4.9,2.7) {$d$};

\draw[-] (5.2,.7) -- (6.2,.7);
\draw[-] (5.2,1.7) -- (6.2,1.7);
\draw[-] (5.2,.7) -- (5.2,1.7);
\draw[-] (6.2,.7) -- (6.2,1.7);

\node at (6.3,1.8) {$c$};


\draw[thick,dotted] (2.5,1.6) -- (7.5,1.6);

\node at (7.7,1.8) {$y = 2 + \epsilon $};

\draw[thick,dotted] (2.5,.6) -- (7.5,.6);

\node at (7.3,.4) {$y = 1 $};

\end{tikzpicture}

\caption{A representation (left) of a $2SIG$ graph (right).}\label{bridge}

\end{figure}

 Notice that, given a representation $R(G)$ of $G$, each such rectangle intersects exactly one stab line partitioning the vertex set 
 $V(G)$ in two disjoint parts, the \textit{lower partition} $V_1$ (vertices with corresponding rectangles intersecting the lower stab line) and 
 the
 \textit{upper partition} $V_2$ (vertices with corresponding rectangles intersecting the upper stab line). 
  Observe that such a vertex partition depends on the representation and is not unique. 
  In the remainder of the article, whenever we speak about a 2SIG with a vertex partition $V(G) = V_1 \sqcup V_2$ 
 we will mean the partitions are lower and upper partition due to a representation. 
 
 Also note that the induced subgraphs $G[V_1]$ and $G[V_2]$ are interval graphs with intervals corresponding to the projection of their rectangles on $X$-axis. 
 Hence, indeed, a 2SIG is obtained by putting some edges between  two different interval graphs. 
Also, observe that the definition of 2SIG does not depend on the specific value of the constant $ \epsilon $ as long as 
it belongs to the interval $(0,1)$. 
Furthermore, observe that a rectangle interval graph with rectangles with unit height locally looks like a 2-stab interval graph.

The article is organized in the following manner. In Section~\ref{preliminaries}
we present the necessary definitions, notations and some observations.
We study the clique number and the chromatic number of 2SIG in Section~\ref{clique number} and justify our claim that 
2SIG and some of its subclasses are ``nearly perfect" even though not perfect.
 We provide a matrix characterization for a subclass of 2SIG graph in Section~\ref{sec matrix}.
Finally, we conclude the article in Section~\ref{conclusion}.

\section{Preliminaries}\label{preliminaries}
The \textit{clique number} $\omega(G)$ of a graph $G$ is the \textit{order} (number of vertices) of the biggest complete subgraph of $G$. 
A \textit{$k$-coloring} of a graph $G$ is an assignment of $k$ colors  to the vertices of $G$ such that adjacent vertices receive different colors.
The \textit{chromatic number} $\chi(G)$ of a graph $G$ is the minimum $k$ such that $G$ admits a $k$-coloring. 
A graph $G$ is \textit{perfect} if $\omega(H) = \chi(H)$ for all induced subgraph $H$ of $G$. 

A graph $G$ is \textit{$\chi$-bounded} if $\chi(H) \leq f(\omega(H))$ for all induced subgraph $H$ of $G$ where $f$ 
is a bounded integer-valued function~\cite{gyarfas}. This is what we meant when we used the informal term ``nearly perfect".

Recall the definition of 2-stab interval graphs from the previous section. 
Now by putting more restrictions on our definition of 2SIG 
we  obtain a few other interesting subclasses of 2SIG that we are going to study in this article.

A  \textit{2-stab unit interval graph (2SUIG)} is a 2SIG with a representation where each rectangle is a unit square. The corresponding representation is 
a \textit{2SUIG representation}. 
A  \textit{proper 2-stab interval graph (proper 2SIG)} is a 2SIG with a 
representation where the projection of a rectangle on $X$-axis does not properly contain the projection of any other rectangle on $X$-asis. 
A  \textit{2-stab independent interval graph (2SIIG)} is a 2SIG with a 
representation where the upper partition induces an independent set. 
The corresponding representation is 
a \textit{2SIIG representation}.


\begin{figure}

\centering
\begin{tikzpicture}

\filldraw [black] (0,0) circle (2pt) {node[left]{}};
\filldraw [black] (2,0) circle (2pt) {node[right]{}};
\filldraw [black] (4,0) circle (2pt) {node[left]{}};
\filldraw [black] (6,0) circle (2pt) {node[right]{}};
\filldraw [black] (8,0) circle (2pt) {node[left]{}};
\filldraw [black] (10,0) circle (2pt) {node[right]{}};

\filldraw [black] (0,2) circle (2pt) {node[left]{}};
\filldraw [black] (2,2) circle (2pt) {node[left]{}};
\filldraw [black] (6,2) circle (2pt) {node[left]{}};
\filldraw [black] (8,2) circle (2pt) {node[left]{}};

\draw[-] (0,0) -- (10,0);

\draw[dashed] (0,0) -- (0,2);
\draw[dashed] (2,0) -- (2,2);
\draw[dashed] (6,0) -- (6,2);
\draw[dashed] (6,0) -- (8,2);

\draw (4,0) .. controls (6,-.8)  .. (8,0);

\end{tikzpicture}

\caption{Example of a bridge triangle free 2SUIG which is also a 2SIIG.}\label{bridgetri}

\end{figure}
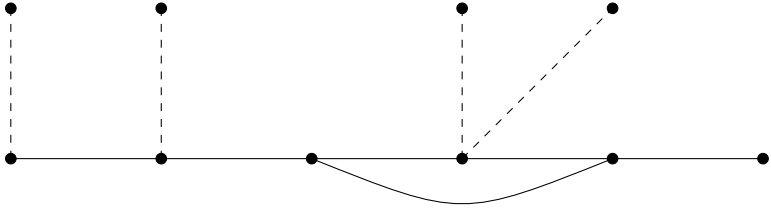

Let us fix a representation $R(G)$ of a 2-stab interval graph $G$ with corresponding lower and upper partitions $V_1$ and $V_2$, respectively.  
Then the set of \textit{bridge edges} $E_B$ is the set of edges (depicted using ``dashed" edges in the figures) 
between the vertices of $V_1$ and $V_2$   while the set \textit{bridge vertices} $V_B$ is the set of vertices incedent to bridge edges (see Fig.~\ref{bridge}). 
For some $v \in V(G)$, the set of \textit{bridge neighbors} $N_B(v)$ is the set of all vertices adjacent to $v$ by a bridge edge.
A \textit{bridge triangle} is a triangle (induced $K_3$) in which exactly two of its edges are bridge edges (note that, no triangle of $G$ can have exactly one or three edges from $E_B$). 
A \textit{bridge triangle free 2SUIG}  is a graph with at least one 2SUIG representation without any bridge triangle (see Fig.~\ref{bridgetri}).

An \textit{orientation} $\overrightarrow{G}$ of a graph $G$ is obtained by 
replacing its edges with \textit{arcs} (ordered pair of vertices). 
An orientation $\overrightarrow{G}$ of $G$ is a \textit{transitive orientation} 
if for each pair of arcs $(a,b)$ and $(b,c)$ we have the arc $(a,c)$ in $\overrightarrow{G}$.
We know that the complement of an interval graph admits a transitive orientation~\cite{mcconne}. 
Let $\overrightarrow{I^c}$ be  a transitive orientation  of the complement graph of an interval graph $I$. 

\section{Relation with other graph classes}

The class of 2SIG graphs can be thought of as a generalization of interval graphs. So, we wondered 
if there is any relation between 2SIG graphs and other generalization of interval graphs, such as,
2-interval graphs. A  \textit{2-interval graph} is a geometric intersection graph where each vertex corresponds to two real intervals~\cite{2interval}.

\begin{proposition}\label{2-int}
All bridge triangle free 2SUIG graphs are 2-interval graphs. 
\end{proposition}

\begin{proof}
Let $G=(V_1\sqcup V_2,E)$ be a bridge triangle free 2SUIG. 
Note that $G[V_1]$ and $G[V_2]$ induces two unit interval graphs. 
We can assign an intervals to each of the vertices of $G$ such that the intersection graph of those intervals is the graph isomorphic to the disjoint union of $G[V_1]$ and $G[V_2]$. These intervals are the first set of intervals assigned to the vertices of $G$.

Now we want to assign a second set of intervals to the vertices of $G$ such that they intersects to represent the remaining edges of $G$. 
Note that the only edges that are not represented yet are exactly the set of bridge edges. 
As $G$ is bridge triangle free 2-stab unit interval graph, the set of bridge edges induces an interval graph isomorphic to disjoint union of paths. Thus, it is possible to assign a second set of intervals, each of them completely disjoint from the intervals belonging to the first set of intervals, to the vertices  such that the intersection graph is isomorphic to the graph induced by bridge edges of $G$.  
\end{proof}

It is well known that proper interval graphs are equivalent to unit interval graphs~\cite{bogart}. 
Interestingly, an analogous result exists for 2SUIG graphs.

\begin{proposition}
The class of proper 2SIG graphs is equivalent to the class of 2SUIG graphs.
\end{proposition}

\begin{proof}
It is easy to observe that a 2SUIG representation of a graph is also a proper 2SIG representation. 

Let $G$ admits a proper 2SIG representation $R$. 
Assume that a vertex $v$ of $G$ corresponds to a rectangle $r_v$. 
Let the projection of $r_v$ 
on $X$-axis be the interval $I_x(r_v)$ and let the projection of $r_v$ on $Y$-axis be the interval $I_y(r_v)$.
Thus, the rectangle $r_v$ is nothing but the cross product $I_x(r_v) \times I_y(r_v)$ of the two intervals.

Consider the intervals obtained from projecting the rectangles on $X$-axis. 
The intersection graph of these intervals will give us a proper interval graph $P$ according to the definition of a proper 2SIG. 
We know that every proper interval graph has a unit interval representation~\cite{golumbic}. Let $U$ be such a representation of $P$. Note that $P$ has all edges 
of $G$ but may have some additional edges as well. Those additional edges $uv$ are precisely those for which 
$I_x(r_u) \times I_x(r_v) \neq \emptyset $   and $I_y(r_u) \times I_y(r_v) = \emptyset $. 
Let $U_v$ be the interval corresponding to a vertex $v$ in $U$. Now consider the rectangle $r'_v = U_v \times I_y(r_v)$ for each vertex $v$ of $G$. Note that these rectangles are unit rectangles and their intersection graph is $G$. Also note that, as we have not changed the $Y$-co-ordinates of the rectangles, the  so obtained representation is still a 2SIG representation. Hence our new representation is indeed a 2SUIG representation of $G$. 
\end{proof}

\section{Clique number and chromatic number}\label{clique number}
A 2SUIG graph is obtained by putting some edges between two interval graphs. The perfectness of an 
interval graph implies that it has chromatic number equal to its clique number. Hence the observation follows.

\begin{observation}\label{uplow}
 Let $G = (V_1 \sqcup V_2, E)$ be a 2SIG graph with a given vertex partition. 
\begin{itemize}
\item[$(i)$] Then $max\{\omega(G[V_1]), \omega(G[V_2])\} \leq \omega(G) \leq \omega(G[V_1]) + \omega(G[V_2])$.

\item[$(ii)$] Then $max\{\chi(G[V_1]), \chi(G[V_2])\} \leq \chi(G) \leq \chi(G[V_1]) + \chi(G[V_2])$.
\end{itemize}
\end{observation}

Let $H$ be a 2SIG with no bridge edges. Then for $H$ both  lower bounds of Observation~\ref{uplow} are tight. 
Now note that even a complete graph with any vertex partition admits a 2SIG representation. In that case, both the 
upper bounds of Observation~\ref{uplow} are tight.   
As the clique number and the chromatic number of an interval graph can be computed in linear time, given a 2SIG with a vertex partition, the lower and upper 
bounds of Observation~\ref{uplow} can be obtained in linear time as well.
As any induced subgraph of a 2SIG is again a 2SIG we have the following result as a direct corollary of the above theorem.

\begin{corollary}
Given any 2SIG graph $G$ we have $\chi(H) \leq 2\omega(H)$ for all induced subgraph $H$ of $G$.
\end{corollary}

\begin{proof}
Let $G$ be a 2SIG with a vertex partition $V(G) = V_1 \sqcup V_2$. As $G[V_i]$ is an interval graph we have 
$\omega(G[V_i]) = \chi(G[V_i])$ for all $i \in \{1,2\}$. Then by Observation~\ref{uplow} we have 

\begin{align}\nonumber
\chi(G) &\leq \chi(G[V_1]) + \chi(G[V_2]) = \omega(G[V_1]) + \omega(G[V_2]) \nonumber \\
&\leq 2max\{\omega(G[V_1]), \omega(G[V_2])\} \leq 2\omega(G).
 \nonumber
\end{align} 
This completes the proof.  
\end{proof}

So in particular 2SIG graphs are $\chi$-bounded which is not surprising as  Gy\'arf\'as~\cite{gyarfas} showed that all boxicity 2 graphs are $\chi$-bounded 
by a quadratic function.  
We showed that 2SIGs are, in fact,  $\chi$-bounded by a linear function. It is known that square intersection graphs are $\chi$-bounded 
by a linear function~\cite{kostochka}.

Now we focus on some of the subclasses of 2SIG. First, we show that 2SIIG graphs are $\chi$-bounded by a better function.

\begin{observation}\label{i2sigcol}
 Let $G = (V_1 \sqcup V_2, E)$ be a 2SIIG graph with a given vertex partition. 
 Then $\omega(G) \leq \chi(G) \leq \omega(G)+1$. 
 
 Moreover, we can enumerate all the maximal cliques and hence, can compute the clique number $\omega(G) $ of $G$ in $O(|V|+|E|)$ time, where $|V|$ is the number of vertices in $G$.
\end{observation} 

\begin{proof}
Note that, the lower partition $V_1$ induces an interval graph and the upper partition $V_2$ induces an independent set. 
Hence $\omega(G[V_1]) = \chi(G[V_1])$ while $\omega(G[V_2]) = \chi(G[V_2]) = 1$. 
Also, Observation~\ref{uplow} implies   $\omega(G[V_1]) \leq \omega(G) \leq \omega(G[V_1]) +1$ and  
$\chi(G[V_1]) \leq \chi(G) \leq \chi(G[V_1]) +1$ which implies this result.  

We know that it is possible to enumerate all the maximal cliques and to compute the clique number of $G[V_1]$ in linear time~\cite{golumbic}. 
For a vertex $v$ in $G[V_2]$, $N_B(v)$ induces an interval graph. The maximal cliques containing $v$ can be enumerated in $O(d)$ time, where $d$ is the degree of the vertex $v$. So, in $O(|V|+|E|)$ time we can compute the clique number of the graph. 
Hence we are done. 
\end{proof}

Note that, if the intersection representation of a boxicity 2 graph, $G=(V,E)$, is given then it is possible to compute the clique number of $G$ in $O(|V|log|V|+|V| \cdot K)$ time, where $K$ is the size of the maximum clique~\cite{snandy}.


Here we provide a 
quadratic time solution for the same problem for 2SIIG, a subclass of boxicity 2 graphs, but we do not require the intersection model as our input in this case. It is enough if the vertex partition of the graph is provied. 
We can prove a similar result for yet another subclass of boxicity 2 graphs, the 2SUIG graphs. The proof is more involved. 

\begin{theorem}\label{allmaxclique}
For any 2SUIG graph $G$ we can enumerate all the maximal cliques and hence, can compute the clique number $\omega(G) $ in polynomial time.
\end{theorem} 

To prove the above result we need to prove the following lemmas.

\begin{lemma}\label{orient}
Let $G = (V_1 \sqcup V_2, E)$ be a 2SUIG graph with a given partition. 
Then there exist transitive orientations $\overrightarrow{G[V_1]^c}$ and $\overrightarrow{G[V_2]^c}$ such that 
for every pair of bridge edges $u_1v_1$ and $u_2v_2$  with $u_1u_2, v_1v_2 \notin E(G)$ 
we have the two arcs 
$(u_1,u_2)$ and $(v_1,v_2)$ in the orientations. 
Moreover, such orientations can be found in polynomial time.
\end{lemma}

\begin{proof}
Take $\overrightarrow{G[V_1]^c}$ and $\overrightarrow{G[V_2]^c}$ such that the statement does not hold. The rectangles corresponding to $u_1$ and $v_1$ along with their intersection divide the region between the axis parallel lines into two disjoint parts. Hence, the intersection between the rectangles corresponding to $u_2$ and $v_2$, cannot be created without any intersection between $u_2$ and $u_1$ or between $v_2$ and $v_1$. This contradicts the premise of the lemma. 
Since, the given graph has representation with the given partition, there must exist $\overrightarrow{G[V_1]^c}$ and $\overrightarrow{G[V_2]^c}$ such that the lemma holds.

Moreover, it is possible to compute such an orientation of $\overrightarrow{G[V_1]^c}$ and $\overrightarrow{G[V_2]^c}$ in $O(|V_1|+|V_2|+|E|)$ time since the transitive orientation of the complement of a connected unit interval graph is unique up to reversal~\cite{Ibarra20091737}.  
\end{proof}

So, given a 2SUIG $G = (V_1 \sqcup V_2, E)$  with a partition we can fix transitive orientations 
$\overrightarrow{G[V_1]^c}$ and $\overrightarrow{G[V_2]^c}$ as in Lemma~\ref{orient}. 
Now given a bridge vertex $v \in V_i$ a vertex $v' \in V_i$ is its 
preceeding bridge vertex if each directed path from $v'$ to $v$ does not go through any other bridge vertex. 
The set 
 of all preceeding bridge vertices of $v$ is denoted by $PBV(v)$.

 Now we will assign integer labels to the bridge edges of $G$. 
 Let $uv \in E_B$ and let
$\mathscr{B}(uv)$ be the set of all bridge edges with  one vertex incident to it 
lying completely to the left (that is, strictly less with respect to the transitive orientation)  of $u$ or $v$.
We assign the integer label $[e]$ to each edge $e= uv \in E_B$ inductively as follows:

\[
    [e] = 
\begin{cases}
    0 & \text{ if PBV(u) = PBV(v)= $\phi$},\\
    i+1 & \text{ otherwise, where $i=max \{ [e']|e' \in \mathscr{B}(e) \}$ }.
\end{cases}
\]

Let $E_{B}^{i}$ be the bridge edges with label $i$.

\begin{lemma}\label{biclique}
In a $2SUIG$ graph the maximal cliques induced by the vertices incident to edges of $E_B^i$ can be enumerated in polynomial time.
\end{lemma}

\begin{proof}
Let $G=(V_1\sqcup V_2,E)$ be a $2SUIG$ graph with a given partition.
Let $G'$ be the subgraph induced  by the vertices incedent to the edges of $E_B^i$ for some fixed index $i$.. 
The vertices of $G'$ belonging to the same stab line create a clique (not necessarily maximal in $G$). 
Consider the subgraph $G'' \subseteq G'$ containing only edges of $E_B^i$. 
Note that, this graph is a bipartite graph. 
Any maximal bipartite clique in $G''$ creates a maximal clique in $G$. 
All maximal bipartite cliques of $G''$ can be enumerated in polynomial time~\cite{Alexe} 
(since we can have at most $O(|V(G'')|)$ maximal bipartite cliques). 
Therefore, the maximal cliques created by the union of the endpoints of $E_B^i$ can be evaluated in polynomial time. 
\end{proof}

Now we will show that it is not possible to have bridge edges with different labels in the same maximal clique of a 2SUIG $G$.

\begin{lemma}\label{noclique}
Bridge edges with different labels are not part of the same maximal clique in a $2SUIG$.
\end{lemma}

\begin{proof}
Let $e,e'$ be two bridge edges and without loss of generality assume $[e] < [e']$. 
Then by definition at least one vertex incident to $e$ is not adjacent to one of the vertices incident to $e'$. Hence we are done. 
\end{proof}

Now we are ready to prove our main result. 

\medskip

\noindent  \textbf{\textit{Proof of Theorem~\ref{allmaxclique}:}} Let $G= (V_1\sqcup V_2,E)$ be a $2SUIG$ graph with a given partition.
We can enumarate all the maximal cliques of $G$ containing at least one bridge edge 
using Lemma~\ref{biclique} and Lemma~\ref{noclique} in polynomial time. 
The maximal cliques of $G[V_1]$ and $G[V_2]$ can be enumerated in polynomial time as they are unit interval graphs. 

\medskip

We could not provide a better $\chi$-bound function for 2SUIG graphs than the one in Observation~\ref{i2sigcol}. 
However, we can provide a better $\chi$-bound function for bridge triangle-free 2SUIG graphs.

\begin{theorem}\label{th trifreecol}
 Let $G$ be a bridge triangle free 2SUIG. Then $\omega(G) \leq \chi(G) \leq \omega(G)+1$. 
\end{theorem} 

However, we will need to prove some lemmas before proving this result.

\begin{lemma}\label{3col}
The bridge vertices of a triangle free 2SUIG graph can be coloured using 2 colors.
\end{lemma}

\begin{proof}
Let $G=(V_1\sqcup V_2,E)$ be a triangle free 2SUIG graph with a given partition and representation. 
Then $G[V_1]$ and $G[V_2]$ are disjoint union of paths. 
Let  $u < v$ if the interval corresponding to $u$ lies in the left of the interval 
corresponding to $v$ (we compare the starting points) 
for any $u,v \in V_i$ where $i \in \{1,2\}$.
Furthermore, we say that  $e=uv < u'v'=e'$ if  
$u<u'$ or 
$v<v'$ where 
$e, e' \in E_B$, 
$u,u' \in V_1$ and 
$v,v' \in V_2$.

We prove the statement using induction on the number of bridge edges.
For $i=1$ the graph is a tree, hence admits a 2-coloring. 

Assume that all bridge triangle-free  2SUIG with at most $k$ bridges admits a 3-coloring such that the bridge vertices 
received only two of the three colors.
Let $G$ be a bridge triangle-free  2SUIG with $k+1$ bridges. 
Let $e'=u'v'$ be an edge of $G$ such that $e< e'$ for all $e \in E_B$. 
Delete $e'$ from $G$ to obtain the graph $G'$. Note that $G'$ admits a 3-coloring where all the bridge edges received 
only two of the three colors.
Let $e''= uv$ be the edge of $G'$ such that $e< e''$ for all bridge edge $e$ of $G'$.  Suppose they received the colors 
$c_1$ and $c_2$. 
 The subgraph induced by the paths $u$ to $u'$ and $v$ to $v'$ is a cycle. If it is an even cycle then we are done. Otherwise, we can always assign the third color to a non-bridge vertex of the cycle and complete the required coloring. 
\end{proof}

Note that the proof of our above result has an algorithmic aspect as well and it is not difficult to observe the following result:

\begin{lemma}\label{trifreecol}
The chormatic number of a triangle free 2SUIG graph can be decided in polynomial time.
\end{lemma}

Now we are ready to prove our main result.

\medskip

\noindent  \textbf{\textit{Proof of Theorem~\ref{th trifreecol}:}}
 Let $G= (V_1\sqcup V_2,E)$ be a bridge triangle free 2SUIG  with a given partition. 
  Note that $G[V_1]$ and $G[V_2]$ are unit interval graphs. 
  We prove the statement using induction on clique number $\omega(G)$.
Note that the theorem  is true for graphs $G$ with $\omega(G)=2$ by Lemma~\ref{trifreecol}.
Assume that the theorem is true for all bridge triangle free 2SUIG $G$ with $\omega(G)\leq m$.
Let $G=(V_1\sqcup V_2,E)$ be a bridge triangle free 2SUIG with $\omega(G)= m+1$. 
We delete a maximal independent set from $G$ to obtain the graph $G'$. 
Note that  $\omega(G') \leq m$ and hence admits a $(m+1)$-coloring by our induction hypothesis.
Now we extend this coloring by assigning a new color to the vertices of the 
deleted maximal independent set to obtain a $(m+2)$-coloring of  $G$.

\medskip

\section{Matrix characterization}\label{sec matrix}
A graph $G$ is a \textit{2-stab unit independent interval graph (2SUIIG)} if it admits a 2SUIG representation where 
the upper partition induces an independent set. 
The corresponding representation is a \textit{2SUIIG representation}.
The corresponding vertex partition $V_1 \sqcup V_2$ is a \textit{strict partition} 
if $G[V_1 \cup \{v\}]$ is not a unit interval graph for any $v \in V_2$ (upper partition). 
 It is easy to see that a graph is a $2SUIIG$ if and only if it has a strict partition.
 For the rest of the section denote the $x$-coordinate of the bottom-left corner of a unit square representing a vertex $v$ by $s_v$.

Note that the class of unit interval graphs is a subclass of 2SUIIG. Moreover, note that the class of 2SUIIG is not perfect as the 5-cycle admits a 
2UIIG representation. Now we characterize the adjacency matrix of a 2SUIIG. 
We will define some matrix forms for that. 
The matrices we consider are 0-1 matrices. The element in the $i^{th}$ row and $j^{th}$ column of a matrix $\mathscr{M}$ is denoted by  
$\mathscr{M}_{ij}$. Also, the $i^{th}$ row and $j^{th}$ column of  $\mathscr{M}$ is denoted by $\mathscr{M}_{i*}$ and $\mathscr{M}_{*j}$, respectively. 
Furthermore,  $First(\mathscr{M}_{i*})$ and $Last(\mathscr{M}_{i*})$ denotes the column indices of the 
 first and last non-zero entries of  $\mathscr{M}_{i*}$, respectively.  

\begin{definition}
A \textit{stair normal interval representation (SNIR) matrix} $\mathscr{A}$  is a 0-1 matrix  with the following properties:

\begin{itemize}
\item[$(i)$] The 1's in a row are consecutive.

\item[$(ii)$] For $j < i$  we have $First(\mathscr{A}_{j*}) \leq First(\mathscr{A}_{i*})$ and 
$Last(\mathscr{A}_{j*}) \leq Last(\mathscr{A}_{i*})$.
\end{itemize}
\end{definition}

Mertzios~\cite{Mer} showed that a graph is a unit interval graph if and only if its adjacency matrix is 
a SNIR matrix. Let $[u \leadsto v]$ denote a longest directed path (not necessarily unique) in $\overrightarrow{I^c}$ 
from $u$ to $v$ and its  length is denoted by $l_{uv}$.


\begin{definition}
A \textit{proper stab adjacency (PSA) matrix} $\mathscr{A}$ is a 0-1 matrix  with the following properties:

\begin{itemize}
\item[$(i)$] The 1's in a row are consecutive and each row has at most two 1's.

\item[$(ii)$] For $j < i$ and $\displaystyle\sum\limits_{k} \mathscr{A}_{jk} =2$ 
we have $First(\mathscr{A}_{j*}) \leq First(\mathscr{A}_{i*})$.

\item[$(iii)$] For $j < i$ and $\displaystyle\sum\limits_{k} \mathscr{A}_{jk} =1$ 
we have $First(\mathscr{A}_{j*}) \leq First(\mathscr{A}_{i*}) + 1$. 
Equality holds only when  $\displaystyle\sum\limits_{k} \mathscr{A}_{ik} =2$.
\end{itemize}
\end{definition}

\begin{definition}
An \textit{independence stair stab representation (ISSR)  matrix} 

\begin{center}
\[ \mathscr{A} _{(m+n) \times (m +n)} = \left[ \begin{array}{c|c}
\mathscr{A'}_{m \times m} & \mathscr{A''}_{m \times n} \\
 \hline
 \mathscr{A''}_{n \times m}^t & 0_{n \times n} \end{array} \right]\]
\end{center}

 is a 0-1 matrix  with the following properties:

\begin{itemize}
\item[$(i)$] The submatrix $\mathscr{A'}$  is a SNIR matrix.

\item[$(ii)$] The submatrix $\mathscr{A''}$ is a PSA matrix.
\end{itemize}

Note that using the characterization given by Mertzios~\cite{Mer} the SNIR submatrix $\mathscr{A'}$ corresponds to a
 a unit interval graph $I$ (say). Let $\overrightarrow{I^c}$ be any transitive orientation of its complement.

\begin{itemize}
\item[$(iii)$] Let $j <i$, $m < k \leq m+n$,  $\mathscr{A}_{ik} = 1$ and $\mathscr{A}_{jk}=1$. Let the rows $\mathscr{A}_{i*}$ and $\mathscr{A}_{j*}$
correspond to the vertices $u$ and $v$, respectively,  of $I$.
Then there is no directed path of length at least two from    $v$ to $u$ in $\overrightarrow{I^c}$. 
\item[$(iv)$] Let $m < k \leq l \leq m+n$,  $\mathscr{A}_{ik} = 1$ and $\mathscr{A}_{jl}=1$. Suppose
$u_1$ and $u_p$ are the vertices corresponding to the rows $\mathscr{A}_{i*}$ and $\mathscr{A}_{j*}$, respectively.
If $[u_1u_2...u_p]$ is a shortest path between $u_1$ and $u_p$ in $I$, then  $(l-k) \leq p + 1$.

\end{itemize}
\end{definition}

\medskip

\noindent Now we are ready to state our main result:

\begin{theorem}\label{matrix}
A graph is a 2SUIIG graph if and only if it can be represented in ISSR matrix form.
\end{theorem}

To prove Theorem~\ref{matrix} we need to prove some lemmas.

\begin{lemma}\label{ubn}
Let $G= (V_1 \sqcup V_2,E)$ be a 2SUIIG with upper partition $V_2$. Then  for a vertex $v \in V_1$ we have $|N_B(v)| \leq 2$.

Moreover, if $N_B(v) = \{x,z\}$ then it is not possible to have a vertex $y \in V_2$ with $s_x < s_y < s_z$. 
\end{lemma}

\begin{proof}
Let $x,y,z \in V_2$ be such that $s_x < s_y < s_z$. As $x,y,z$ are independent, we must have $s_x +1 < s_y < s_y +1 < s_z$. 
If a vertex $v \in V_1$ is adjacent to both $x$ and $z$ then we must have $s_v \leq s_x +1$ and $s_z \leq s_v +1$. 
This implies $s_v \leq s_x +1 < s_y < s_z -1 \leq s_v$, a contradiction. 
 \end{proof}

\begin{lemma}\label{ubn_clique}
Let $G= (V_1 \sqcup V_2,E)$ be a 2SUIIG with upper partition $V_2$. Then for two vertices $u,v \in V_1$ 
we have $\left|\displaystyle\bigcup\limits_{k=1}^{p} N_B(u_k)\right| \leq p+1$ where 
$u_k$ are the vertices of the shortest path $(P)$ between $u,v$ in $G[V_1]$ and
$p$ is the length of $P$.
\end{lemma}

\begin{proof}
Let $N=\displaystyle\bigcup\limits_{k=1}^{p} N_B(u_k)$ and $|N| > p+1$.
Let there is a representation $R$ of $G$.
Without loss of generality assume $s_u < s_v$ in $R$.
There must be a vertex $y,z\in N$ such that $(z,u),(y,v) \in E$,
$s_u-1<s_z<s_u$ and 
$s_z+p+1+\epsilon$, $0<\epsilon<1$. 
In any representation of $G$, $|s_u - s_v| $ is at most $p-\epsilon$,$0<\epsilon<1$.
Then the edge between $y,v$ cannot be realised leading to a contradiction.
\end{proof}

The above lemmas gives an upper bound on the number of bridge neighbours of a vertex in $V_1$.

\begin{lemma}\label{lbc2}
Let $G= (V_1 \sqcup V_2,E)$ be a 2SUIIG with upper partition $V_2$. 
For vertices $u,v \in V_1$ with $|N_B(u)| = 2$ and $s_u<s_v$ there exists $w\in N_B(u)$ such that for any
 $x \in N_B(v)$ we have $s_w \leq s_x$.
\end{lemma}

\begin{proof}
Let there are vertices $w,x \in V_2$ such that $s_x < s_w$ where $w\in N_B(u)$ and $x \in N_B(v)\setminus N_B(u)$. 
Then the union of the unit squares corresponding to $u$ and $w$ divides the region enclosed by the stab lines 
into two disjoint parts. 
Then the intersection of the unit squares corresponding to $v$ and $x$ cannot be realised.  
\end{proof}

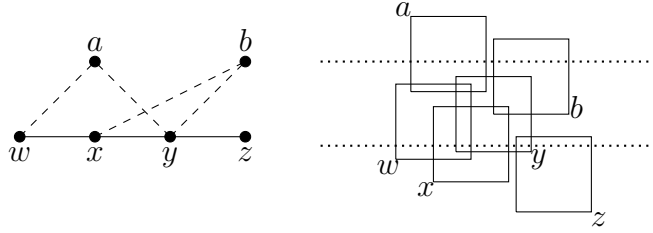
\begin{figure}

\centering
\begin{tikzpicture}

\filldraw [black] (0,0) circle (2pt) {node[below]{$w$}};
\filldraw [black] (1,0) circle (2pt) {node[below]{$x$}};
\filldraw [black] (2,0) circle (2pt) {node[below]{$y$}};
\filldraw [black] (3,0) circle (2pt) {node[below]{$z$}};

\filldraw [black] (1,1) circle (2pt) {node[above]{$a$}};
\filldraw [black] (3,1) circle (2pt) {node[above]{$b$}};

\draw[-] (0,0) -- (3,0);

\draw[dashed] (0,0) -- (1,1);
\draw[dashed] (2,0) -- (1,1);
\draw[dashed] (1,0) -- (3,1);
\draw[dashed] (2,0) -- (3,1);


\draw[-] (5,-.3) -- (6,-.3) -- (6,.7) -- (5,.7) -- (5,-.3);

\node at (4.9,-.4) {$w$};

\draw[-] (5.5,-.6) -- (6.5,-.6) -- (6.5,.4) -- (5.5,.4) -- (5.5,-.6);

\node at (5.4,-.7) {$x$};

\draw[-] (5.2,.6) -- (6.2,.6) -- (6.2,1.6) -- (5.2,1.6) -- (5.2,.6);

\node at (5.1,1.7) {$a$};

\draw[-] (5.8,-.2) -- (6.8,-.2) -- (6.8,.8) -- (5.8,.8) -- (5.8,-.2);

\node at (6.9,-.3) {$y$};

\draw[-] (6.3,.3) -- (7.3,.3) -- (7.3,1.3) -- (6.3,1.3) -- (6.3,.3);

\node at (7.4,.4) {$b$};

\draw[-] (6.6,-1) -- (7.6,-1) -- (7.6,0) -- (6.6,0) -- (6.6,-1);

\node at (7.7,-1.1) {$z$};


\draw[thick,dotted] (4,-.12) -- (8.5,-.12);

\draw[thick,dotted] (4,1) -- (8.5,1);

\end{tikzpicture}

\caption{The scenario dealt with in Lemma~\ref{lbc1}.}\label{fig adjacency}

\end{figure}

\begin{lemma}\label{lbc1}
Let $G= (V_1\sqcup V_2,E)$ be a 2SUIIG with upper partition $V_2$. 
For vertices $x,y \in V_1$ with $|N_B(x)| = 1$ and $s_x<s_y$ we have $s_b - 2 < s_a$ where $b\in N_B(x)$ and $a \in N_B(y)$.
\end{lemma}

\begin{proof}
Let we have vertices $x,y \in V_1$ with $|N_B(x)| = 1$ and $s_x<s_y$ we have $s_a \leq s_b - 2$ 
where $b\in N_B(x)$ and $a \in N_B(y)$. Now $s_b < s_x+1$. If $s_a \leq s_b - 2$ 
then to realize the intersection between $y,a$ we will have $s_y < s_x$, 
which is a contradiction. Note that the bound $s_b - 2$ is tight as 
we can have a situation illustrated in Fig.~\ref{fig adjacency}. 
Moreover in such situations $b$ must be adjacent to $y$ also. 
\end{proof}

Intuitively, Lemma~\ref{lbc2} and Lemma ~\ref{lbc1} show that in a 2SUIIG representation, 
the ordering of the intervals corresponding to the vertices in $G[V_1]$ fixes an ordering of the intervals corresponding to 
the vertices in $G[V_2]$.

Let $G= (V_1\sqcup V_2,E)$ be a 2SUIIG with upper partition $V_2$. Assume that $|V_1| = m$ and $|V_2| = n$
We know that it admits a strict partition.  
From this strict partition we will construct a matrix which we will prove is a PSA matrix.
Note that $G[V_1]$ is a unit interval graph and hence Mertzios~\cite{Mer} provides a (adjacency) SNIR matrix $\mathscr{A}_{m\times m}$ 
of $G[V_1]$. 
This matrix is obtained by putting the vertices of $V_1$ in a particular order.
We will show that, we can obtain a particular order of $V_2$ such that the biadjacency matrix 
$\mathscr{A}^{'}_{m\times n}$ of $V_1$ (taken in the same order as above) and $V_2$ is a PSA matrix.

Call it $\mathscr{A}_{m\times m}$. $\mathscr{A}^{'}_{m\times n}$ be a zero one matrix. 
An entry $\mathscr{A}^{'}_{ij}$ is 1 if there is a bridge between $i,j$.
The ordering of the rows of $\mathscr{A}^{'}$ remains same as the ordering of the rows in $\mathscr{A}$.
The ordering of the columns of $\mathscr{A}^{'}$ corresponds to an ordering of the vertices of $V_2$.

We order the vertices of $V_2$ using the following prescribed rule. 
Recall that the adjacency matrix $\mathscr{A}$ of $G[V_1]$ produces an interval intersection
 representation of (may not be unique) 
the graph.
Fix one such representation $R$.
Suppose we have $u,v \in V_1$ and $u',v' \in V_2$ such that 
$u' \in N_B(u)\setminus N_B(v)$ and $v' \in N_B(v)$.
If $u,v$ are not twins (that is, have the same set of neighbors) in $G[V_1]$ and $s_u < s_v$ in $R$ then
we want  $u'< v'$
in our ordering of $V_2$.
In any other case we take an arbitrary ordering between a pair of vertices  $u',v' \in V_2$.
This is a well defined algorithm for getting the ordering due to the previously proved lemmas.

\begin{lemma}\label{PSA}
The matrix $\mathscr{A}^{'}$ is a PSA matrix.
\end{lemma}

\begin{proof}
Suppose $\mathscr{A}^{'}$ is not PSA matrix. Note that, there cannot be any zero column in it as we have assumed a strict
partition of the graph. Now we will check all the properties mentioned in the definition of the PSA matrix.

\medskip

\textit{Property $(i)$:} Due to Lemma~\ref{ubn} each vertex in $V_1$ have at most two neighbours 
and they are consecutive. 

\medskip

\textit{Property $(ii)$:} Suppose a vertex in $v \in V_1$ has two bridge neighbours and $w$ be the left most bridge
neighbour of $v$. $u$ be another vertex in $V_1$ lying to the right of $v$. 
From Lemma~\ref{lbc2}, we know that the bridge neighbours of 
$u$ will lie to the right of $w$. Hence, this property is also satisfied. 

\medskip

\textit{Property $(iii)$:}  Assume that this property is not satisfied and $x,y$ be the
vertices corresponding to the rows violating the property.
Now $s_x < s_y$. Let $b$ be the bridge neighbour of $x$.
There is a bridge neighbour $a$ of $y$
such that there is another vertex $c \in V_2$ with $s_a < s_c < s_b$. 
Hence, $s_b - 2 > s_a$.  This contradicts Lemma~\ref{lbc1}. $\hfill \square$
\end{proof}

Now we are ready to prove our main theorem.

\medskip

\noindent  \textbf{\textit{Proof of Theorem~\ref{matrix}:}}
First we will prove the ``if" part. 
Let $G= (V_1\sqcup V_2,E)$ be a 2SUIIG with upper partition $V_2$.
Then by the discussions above Lemma~\ref{PSA} we obtained an ordering of the vertices of 
$V_1$ and $V_2$. Consider the ordering of $V_1 \sqcup V_2$ by putting the vertices of $V_1$ 
in the previously obtained order followed by the vertices of $V_2$ (in previously obtained order as well).
This ordering will give us a matrix of the following type:

\begin{centering}

\[ Adj(G) = \left[ \begin{array}{c|c}
 \mathscr{A} & \mathscr{A'} \\
 \hline
 \mathscr{A'}^t & 0 \end{array} \right]\]

\end{centering}
 
From the previous lemmas we know that the above matrix satisfies the first two properties for being a ISSR matrix.
Now we check the remaining two properties.

\textit{Property $(iii)$:} As a consequence of Lemma~\ref{ubn} any vertex in $V_2$ can have at most two bridge neighbours 
which are pair wise independent.

\medskip

\textit{Property $(iv)$:} As a consequence of Lemma~\ref{ubn_clique} for any vertices in $G[V_1]$,
the vertices of the shortest path between $u,v$ in $G[V_1]$
can have at most a total of $p+1$ bridge neighbours where $p$ is the length of the shortest path.

For the ``only if'' part, given a ISSR matrix we can 
construct a unit interval graph (and generate the intervals corresponding to the vertices) with the SNIR sub-matrix of it.
Now we show that the column ordering of the PSA sub-matrix generates the intervals of the rest of the vertices.
If the PSA sub-matrix have only one non-zero cell, then there is only one vertex, $x$. 
Due to property (iii) of ISSR matrix the unit square $s_x$ can be generated 
and the bridge edge can be realized. Assume that the theorem is true for all
ISSR matrix whose PSA sub-matrix have $k$ columns. Consider a ISSR matrix $M$
whose PSA sub-matrix have $m=k+1$ columns. Consider $M'$ as the matrix 
which is obtained by deleting the last columns of the PSA sub-matrix of $M$.
From induction hypothesis, there is a 2SUIIG representation of $M'$, say $R$.
Let $x$ be a vertex corresponding to the $(k+1)^{th}$ column.
Let $Y$ be the set of vertices corresponding to the rows $i$ such that
$M_{i(k+1)} = 1$. $y$ be an element of $Y$. So, in any 2SUIIG representation
of $M$, $|s_y - s_x | \leq 1$.

\textbf{case 1.} Let there is a vertex $y \in Y$ such that for any $s_x$ with 
$s_y < s_x < s_y + 1$ in $R$, $G[V_2]$ remains independent but 
the bridge edge $xy$ cannot be realised.
The bridge edge $xy$ cannot be realised implies there are $u,v$ in $V_2,V_1$ 
respectively such that $s_u < s_x,s_y < s_v$ and $u,v$ are adjacent.
If $y,u$ were adjacent, then we could have realised the bridge edge between $x,y$.
Again if $v,x$ were not adjacent then $M$ would violate property (iii) of the
definition of PSA matrix. The adjacency of $v,x$ ensures that the bridge edge between $x,y$
can be realised.

\textbf{case 2.} Let there is a vertex $y \in Y$ such that for any $s_x$ with 
$s_y < s_x < s_y + 1$ in $R$, $G[V_2]$ does not remain independent.
This implies that there exists $u,v$ in $R$
such that if $s_y < s_x < s_y + 1$ and $G[V_2]$ is an independent set
then the bridge edge between $u,v$ 
cannot be realised.
Let $Z=~\{a:~a~\in~V_2~\text{ and }~s_u-1~< s_a <~s_y+1~\text{ in R}\}$. 
Let $P=[u_1u_2\ldots u_k]$ be the shortest path from $u$ to $y$ in $G[V_1]$
where $u_1=u$ and $u_k=y$. 
Now $ |Z| = k+1$, otherwise the bridge edge between $x,y$ can be realised.
But then $M$, violates the definition of ISSR matrix.

Property (iii) of ISSR matrix insures that for all $y,z \in Y$ with $s_y < s_z$ in $R$,
if $s_y < s_x < s_y+1$ we can have $|s_z - s_x | \leq 1$. Hence, all the bridge edges
between $x$ and the elements of $Y$ can be realised.
This completes the proof.

\section{Conclusion}\label{conclusion}
The complexity of recognizing 2SIG is not known and hence is an interesting future problem. 
Domination number of 2SIG is one of the prospective areas of research as
the problem is polynomial time solvable for interval graphs but is $NP$-hard for boxicity 2 graphs.
One can also generalize the concept 
to define $k$-stab interval graphs and study its different aspects.




\bibliographystyle{abbrv}
\bibliography{science}

\end{document}